\documentclass[10pt]{article}
\usepackage{amsmath,amsthm,latexsym,amssymb,amsfonts,parskip,epsfig,xspace}
\usepackage[a4paper,top=2.5cm,bottom=3.5cm]{geometry}

\usepackage{palatino}
\usepackage[mathcal]{euler}

\numberwithin{equation}{section}
\allowdisplaybreaks

\newtheorem{lemm}{Lemma}[section]
\newtheorem{prop}[lemm]{Proposition}

\newcommand{\R}{\mathbb{R}}                  

\newcommand{\scpr}[2]{\langle#1\, \vert \, #2 \rangle}

\newcommand{\betr}[1]{\left\lvert #1 \right\rvert}

\newcommand{\mal}{\mu_{\text{AL}}}

\newcommand{\wt}[1]{\widetilde{#1}}
\newcommand{\hilb}{\mathcal{H}}
\newcommand{\bg}[1]{#1^{(0)}}
\newcommand{\pin}{\pi'}

\newcommand{\sscpr}[3]{\langle#1\, \vert \, #2 \, \vert \, #3\rangle} 
\newcommand{\ket}[1]{\lvert\, #1\rangle}
\newcommand{\bra}[1]{\langle\, #1\rvert}

\newcommand{\vac}{\text{vac}}

\DeclareMathOperator{\abar}{\overline{\mathcal{A}}}
\DeclareMathOperator{\tr}{tr}

\DeclareMathOperator{\one}{\mathbb{I}}
\DeclareMathOperator{\id}{id}

\DeclareMathOperator{\tdiff}{TDiff}
\DeclareMathOperator{\diff}{Diff}
\DeclareMathOperator{\symm}{Symm}
\DeclareMathOperator{\tgiff}{TGiff}
\DeclareMathOperator{\giff}{Giff}
\DeclareMathOperator{\gymm}{Gymm}

\DeclareMathOperator{\gauge}{\mathcal{G}}
\DeclareMathOperator{\gont}{\triangleright}
\DeclareMathOperator{\gone}{\triangleright}
\DeclareMathOperator{\pone}{\triangleright}
\DeclareMathOperator{\pont}{\triangleright}

\DeclareMathOperator{\ad}{ad}

\title{On loop quantum gravity kinematics with non-degenerate spatial
background}
\author{Hanno Sahlmann\\{\small Institute for Theoretical Physics, Karlsruhe
Institute of Technology, Karlsruhe (Germany)}}
\date{{\small Preprint KA-TP-14-2010}}

\begin{document}
\maketitle
\begin{abstract}
In a remarkable paper, T.\ Koslowski introduced kinematical representations for
loop quantum gravity in which there is a non-degenerate spatial background
metric present. He also considered their properties, and showed that Gau{\ss}
and  diffeomorphism constraints can be implemented. With the present article, we
streamline and extend his treatment.
In particular, we show that the standard regularization of the geometric
operators leads to well defined operators in the new representations, and we
work out their properties fully. We also give details on the implementation of
the constraints. All of this is done in such a way as to show that the standard
representation is a particular (and in some ways exceptional) case of the more
general constructions. This does not mean that these new representations are 
as fundamental as the standard one. Rather, we believe they might be useful to 
find some form of effective theory of loop quantum gravity on large scales.
\end{abstract}
\section{Introduction}
In the paper \cite{Koslowski:2007kh}, T.\ Koslowski introduced
kinematical representations for loop quantum gravity in which there is a
non-degenerate spatial background
metric present. In these representations, the flux-operators of loop quantum
gravity (encoding quantized spatial geometry) acquire a c-number term in
addition to the standard one:
\begin{equation}
\label{eq_nrep}
\pi(E_{S,f})={X}_{S,f} +  \bg{E}_{S,f}\id, \quad \text{ with }  
\bg{E}_{S,f}=\int_S *\bg{E}{}^I f_I.
\end{equation}
The quantity $\bg{E}_{S,f}$ is exactly the classical value of the flux in a 
background geometry given by a densitized triad $\bg{E}$. The possibility of
such representations was known to some experts before
\cite{Koslowski:2007kh} (they were mentioned for example in passing in 
\cite{Sahlmann:2002xv}) but they were not considered interesting, because
they did not alleviate the asymmetry in fluctuations between the canonical
variables, and because gauge transformations and diffeomorphisms could not be
implemented unitarily. The remarkable discovery of \cite{Koslowski:2007kh} is
that the transformations \emph{can} be implemented if one is willing to go to a
large direct sum of such representations, and that, even better, the
corresponding constraints can then be implemented.\footnote{It may be
interesting to compare this to the situation in \cite{Varadarajan:2007dk},
where another new and interesting representation for loop quantum gravity is
constructed, which is likewise highly reducible.} 

The goal of the present article is to round off the picture that is emerging in 
\cite{Koslowski:2007kh}, by showing that geometric operators can be defined in
the new representations \emph{exactly} as in the standard representation, and
that they have a very simple structure. For example we find for the volume
operator of a region $R$ that 
\begin{equation}
 V_R= V_R^\vac+\bg{V}_R\id 
\end{equation}
where $\bg{V}_R$ is the classical volume of $R$ in the spatial background
geometry. Additionally we discuss the implementation of Gau{\ss} and
diffeomorphism constraint, using exactly the same techniques as in the standard
representation. This adds some detail and clarifications to the treatment in
\cite{Koslowski:2007kh}. Finally \cite{Koslowski:2007kh} was using an algebra
for loop quantum gravity that many researchers may not be very familiar with.
Here we will work with the standard formalism.

As for the significance of these representations, we believe that they are not
as fundamental as the standard representation. The latter shows how geometry
can emerge from literally nothing, and it beautifully explains the black hole
entropy-area relation, and more. The new representations, by contrast, are
based on a classical geometry as input. Still they may be useful in some form of
perturbative calculation, in which a huge number of elementary geometric
excitations are subsumed into the classical background. 
 
We should point out that the new representations are treating the two
canonically conjugate quantities $A$ and $E$ very differently, in much the same
way as the standard representation does. The states in the new representations,
when viewed as functions of $E$, are sharply peaked around a classical spatial
geometry, whereas they are almost constant when viewed as functions of the
conjugate variable $A$. In this sense they fail to approximately describe a
classical \emph{space-time}. They only approximately describe a classical
\emph{spatial} geometry. While it would be very interesting to distribute the
quantum mechanical uncertainty more evenly, and thus introduce also a classical
background connection, the states considered here do not do that, and it is also
not straightforward to generalize them in such a way that they do. It is however
conceivable be that they can be taken as the ingredients of some form of
limiting procedure, to obtain states that encode a full space-time geometry.   

We should also remark that the new representations, as well as those obtained
from them for the purpose of unitarily implementing the spatial diffeomorphisms 
are \emph{not} counterexamples to the uniqueness theorems
\cite{Fleischhack:2004jc,Lewandowski:2005jk} on representations of the
kinematical algebra. In both cases, crucial hypotheses are violated: As we will
see in the next section, there is no unitary implementation of the
the action of spatial diffeomorphisms on the basic variables in the new
representations. And in the larger spaces that do admit such a unitray
implementation, the representation is highly reducible. 

In the next section, we recall the definition of the new representations from
\cite{Koslowski:2007kh} and state some elementary properties. Then we
investigate the geometric operators in section \ref{se_geo}. In sections
\ref{se_uni} and \ref{se_imp} we discuss implementations of
gauge transformations and diffeomorphisms, and of the constraints. We finish
with some discussion and outlook.

\section{The new kinematical representations}
To fix notation, we will start by reviewing the standard LQG quantum kinematics.
The canonical pair to be quantized consists of an SU(2) connection $A$ and a
densitized triad field $E$. These fields transform under the automorphisms of
the SU(2) bundle. Upon choosing (local) trivializations, these bundle
automorphisms split into gauge transformations and diffeomorphisms. The
automorphism group is then the semi-direct product of those subgroups, 
\begin{equation}
\label{eq_dirpr}
(g,\phi) \cdot (g',\phi') =(g a_\phi(g'),\phi\circ\phi'),
\end{equation}
with $g,g'$ gauge transformations, $\phi,\phi'$ diffeomorphisms, and the
automorphism $a_\phi(g)= \phi_* g$ on the group of gauge transformations given
by the push-forward under diffeomorphisms. 
The fields then transform under gauge transformations
$g$ and diffeomorphisms $\phi$ in the standard way:
\begin{gather}
 A \overset{g}{\longmapsto} \ad_g(A)+g^{-1}dg, \qquad A
\overset{\phi}{\longmapsto} \phi_{*}A\\
 E \overset{g}{\longmapsto} ad_g(E) =:g \gone E, \qquad E
\overset{\phi}{\longmapsto} \phi_{*}E=:\phi \pone E
\end{gather}
where $\ad$ is the adjoint action of SU(2) on su(2), and the star signifies
push-forward. The basic variables used for quantization are chosen in such a way
as to
make their transformation behavior under SU(2) and spatial diffeomorphisms as
simple and transparent as possible. For the connection, one considers 
holonomies
\begin{equation}
h_{\alpha}[A]=\mathcal{P} \exp \int_{\alpha} A, 
\end{equation}
or more generally, functions of such holonomies, 
\begin{equation}
\label{eq_cyl}
T[A]\equiv t(h_{\alpha_1}[A],h_{\alpha_2}[A],\ldots,h_{\alpha_n}[A])
\end{equation}
for a finite number of paths $\alpha_1,\ldots, \alpha_n$ forming a graph, and a
function $t$ on coupys of SU(2). Such
functionals are also called \emph{cylindrical functions}. 

For the field $E$ a natural functional is its flux through surfaces $S$:
\begin{equation}
\label{eq_flux}
E_{S,f}[E]=\int_S *E_I f^I
\end{equation}
where $f$ is a function taking values in su(2)$^*$ and $*E$ is the two-form 
$E^a \epsilon_{abc}\text{d}x^b \wedge \text{d}x^c$. 

The action of gauge transformations and diffeomorphisms on $(A,E)$ induces a
representation on cylindrical functions and fluxes. It is given by
\begin{align}
\label{eq_act1}
&(g\gont T)[A]=
t(g(s1)^{-1}h_{\alpha_1}g(e1), g(s2)^{-1}h_{\alpha_2}g(e2),\ldots)\\
&(\phi\pont T)[A]=
t(h_{\phi(\alpha_1)}[A],h_{\phi(\alpha_2)}[A],\ldots)\\
&(g\gone E_{S,f})=E_{S,\ad_g(f)}\\
&(\phi\pone E_{S,f})=E_{\phi(S),\phi_*f} 
\label{eq_act4}
\end{align}
with $t$ the function on SU(2)$^n$ characterizing the cylindrical functional
$T$ according to \eqref{eq_cyl}. 
To quantize cylindrical functions and fluxes, one is seeking a representation of
the following algebraic relations on a Hilbert space: 
\begin{equation}
\label{eq_qalg}
\begin{split}
 T_1\cdot T_2[A]&=T_1[A]T_2[A]\\
  [T,E_{S,f}]&= 8\pi \iota l^2_P X_{S,f}[T]\\
  [T,[E_{S_1,f_1}, E_{S_2,f_2}]]&=(8\pi \iota l^2_P)^2[X_{S_1,f_1},
X_{S_2,f_2}][T]  \\
  \ldots &\\
  (E_{S,f})^*=E_{S,\overline{f}},\quad&\quad
  (T[A])^*=\overline{T}[A]
  \end{split}
\end{equation}
Here, $X$ is a certain derivation on the space of cylindrical functions. These
relations define an algebra. Diffeomorphisms and gauge transformations act as
automorphisms. 

The kinematical Hilbert space
\begin{equation}
 \hilb_\text{kin}=L^2(\abar, \mal). 
\end{equation}
is a space of functionals over the space of generalized connections $\abar$. The
fluxes are represented by derivatives $X_{S,f}$, 
\begin{equation}
 \pi(E_{s,f})=X_{S,f}. 
\end{equation}
Cylindrical functions act by multiplication. Gauge transformations and 
diffeomorphisms (and their semi-direct product) are unitarily implemented  
by the operators 
\begin{equation}
U_\phi \ket{T}=\ket{\phi \pont T}, \qquad U_g \ket{T}= \ket{g\gont T}.
\end{equation}
This representation is called \emph{Ashtekar-Lewandowski representation}. In
the present article, we will also use the term \emph{vacuum representation}, to
contrast it with the new representations, which can be thought of as
representations containing a ``geometry condensate'' \cite{Koslowski:2007kh}.

Now we turn to the new representations \cite{Koslowski:2007kh}. Let $\bg{E}$ be
a classical triad field. Then, following Koslowski we define the new
representation on \emph{the same} Hilbert space as
before, by changing the action of the fluxes:
\begin{equation}
 \pin(E_{S,f})={X}_{S,f} +  \bg{E}_{S,f}\id, \quad \text{ with }  
\bg{E}_{S,f}=\int_S *\bg{E}{}^I f_I.  
\end{equation}
It is easily checked that this gives another representation of the algebra. We
list some elementary properties. 
\begin{enumerate}
 \item The spectra of the fluxes have changed. If $\lambda$ is an eigenvalue
of $X_{S,f}$ then $\lambda+ \bg{E}_{S,f}$ is an eigenvalue of $\pin(E_{s,f})$. 
In particular, a constant
cylindrical function is now an eigenfunction of $\pin(E_{s,f})$ with in general
non-vanishing eigenvalue $\bg{E}_{S,f}$.
\item The new representation is still cyclic, with the empty spin net as cyclic
vector.
\item The new representation is unitarily inequivalent to the standard one.  
\item The operators implementing diffeomorphisms and gauge transformations in
the standard representation are still well defined and unitary in the new
representation, but it can be easily checked that they do not implement the
algebra-automorphisms anymore. For example
\begin{equation}
 U_\phi \pin(E_{S,f}) U_{\phi^{-1}}\neq \pin(\phi \pone E_{S,f})
\end{equation}
in general. The reason is that the $U_\phi$ do not change the background
geometry.\footnote{Indeed, one finds 
\begin{equation*}
 U_\phi \pin(E_{S,f}) U_{\phi^{-1}}=
U_\phi \left(X_{S,f}+ \bg{E}_{S,f}\one \right) U_{\phi^{-1}}
= X_{\phi(S),\phi_*f} + \bg{E}_{S,f}\one
\end{equation*}
by virtue of the fact that the $U_\phi$ act on the $X_{S,f}$ in the standard
way, but commute with the new c-number term. 
}  

\item The standard kinematical representation can be viewed as a special case
of the new one, for $\bg{E}=0$. 
\end{enumerate}
We note that point 5 is an empediment to the implementation of the
gauge and diffeomorphism constraints. It is also the reason that the new
representations are not counterexamples to the uniqueness
theorems \cite{Fleischhack:2004jc,Lewandowski:2005jk}, since the latter require
the implementation of the action \eqref{eq_act1} -- \eqref{eq_act4} on the
basic operators. 

We note again that in the new representations, the operators representing the
connection remain untouched. One might wonder whether it is possible to obtain
representations in which also a background connection is used. A straightforward
way to proceed would be as follows:\footnote{The following argument is due to J Lewandowski (private communication).} The new representations from above can be
obtained by applying a translation 
\begin{equation}
E\mapsto E+\bg{E} 
\end{equation}
to the functionals $E_{S,f}$ of \eqref{eq_flux}.
A similar translation, 
\begin{equation}
A\mapsto A+\bg{K}, 
\end{equation}
is possible on the connection, where $\bg{K}$ is a suitable fixed tensor field.
The difference is that the translation on the triad field, when applied to the
functional $E_{S,f}$, leads to a functional of the same type (flux plus constant
term), whereas the translation on the connections turns a holonomy into a very
complicated functional of the connection. Thus this translation leads out of the
algebra of functionals that is the basis for loop quantum gravity and can thus
not be used to define new states.   

In the following section we will consider the existence and properties of the
geometric operators for area and volume in the new representation. We then turn
to the implementation of diffeomorphisms and gauge transformations, and to the
implementation of the corresponding constraints.
\section{Geometric operators}
\label{se_geo}
In this section, we will show that area and volume operators can be defined in
the new representations with exactly the same regularization and quantization
strategy (\cite{Rovelli:1994ge,Ashtekar:1996eg,Ashtekar:1997fb}) as in the
vacuum representation. This simplifies the treatment in \cite{Koslowski:2007kh},
in which a new volume operator was used. Also, we extend the results about
expectation values of the area operator obtained there. In fact, we make the
structure of area and volume operator completely explicit. It is very simple.
The action of the geometric operators in the new representations is given as a
sum
of two terms. One term is precisely the action in the vacuum representation, the
other one is given as the value of the geometric quantity in the background
geometry $\bg{E}$ times the identity operator.
\begin{prop}
The standard regularization and quantization procedures for area and 
volume operators can be applied to the new representations and lead to well
defined operators. We find  
\begin{equation}
 V_{R}=V_R^{\vac}+ \bg{V}_R\id, \qquad A_S=A^{\vac}_S+\bg{A}_S\id
\end{equation}
with $V^{\vac}, A^{\vac}$ the geometric operators in the vacuum representation,
and 
$\bg{V}_R,\bg{A}_S$ the classical values in the background geometry. 
\end{prop}
This result has a very simple form, but it is by no means trivial. A priori,
coupling terms between geometry excitations and background could have appeared,
or the regularization procedures could have broken down completely. We will have
to work somewhat to exclude both possibilities. The result is remarkable in
that it suggests a split into a background part (the c-number term) and
fluctuations (represented by the vacuum operators), and thus a physical
interpretation of the states in the new representations as excitations above a
background spatial geometry. It should be stressed, however, that this split is
a result of the application of the standard quantization procedure
\cite{Rovelli:1994ge,Ashtekar:1996eg,Ashtekar:1997fb} to the new
representations, not an assumption. 

As we have already done in
stating the result, from now on we will drop all the $\pi$'s indicating the
representation. We hope that it is clear from context. Sometimes, we may add an
index $\vac$ to denote an object in the vacuum representation, for clarity. 
\subsection{Area}
Let us first consider the area operator. The whole calculation is modeled on
\cite{Ashtekar:1996eg}, so we refer to that reference for further background. 
We pick a surface $S$ and consider a family of non-negative densities
$f^\epsilon_p(q)$ on
$S$ that has the property
\begin{equation}
\lim_{\epsilon\rightarrow 0} f^\epsilon_p(q)=\delta_p(q).
\end{equation}
$\delta_p$ is the delta-function on $S$ which is peaked at $p$. We also
introduce the shorthands 
\begin{equation}
E^I_{\epsilon}(q)=\pi(E_{S,b^If^\epsilon_p(\cdot)}), \qquad
X^I_{\epsilon}(q)=X_{S,b^If^\epsilon_p(\cdot)}
\end{equation}
and 
\begin{equation}
\bg{E}{}^{I}_{\epsilon}(p)=\int_S*\bg{E}{}^If^{\epsilon}_p.
\end{equation}
Then we can define 
\begin{equation}
\Delta_\epsilon=E_\epsilon\cdot E_\epsilon. 
\end{equation}
As in \cite{Ashtekar:1996eg}, we would like to define
\begin{equation}
\label{eq_area}
\widehat{A}_S=\lim_{\epsilon\rightarrow 0} \int_{S} \sqrt{\Delta_\epsilon} 
\end{equation}
and study its properties. It can be seen immediately that $\Delta_\epsilon$ is
essentially self adjoint on cylindrical functions, and positive semidefinite.
Therefore taking its square root is not a problem. Rather, what we have to worry
about is the limit in \eqref{eq_area}. The structure of $\Delta_\epsilon$ is 
\begin{equation}
\Delta_\epsilon=
X\cdot X+2X\cdot\bg{E}+\bg{E}\cdot\bg{E}
\end{equation}
where the first term is the standard term. To see what the additional terms
give, the first step is to go to a basis of eigenvectors for both, the
${X}\cdot{X}$ and the ${X}\cdot\bg{E}$ term. We can
consider the case of a single intersection vertex $v$ (see figure
\ref{fi_area}), as the more general case can be handled by first subdividing the
surface and considering the individual area operators separately.
\begin{figure}
\centerline{\epsfig{file=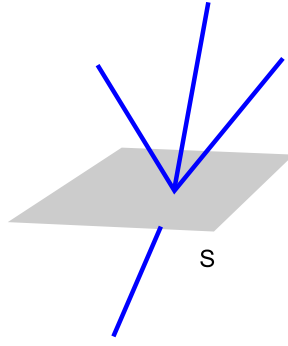, scale=0.7}}
\caption{\label{fi_area} Surface with single vertex $v$ as described in the
text.}
\end{figure}
${X}\cdot{X}$ and ${X}\cdot\bg{E}$ can be diagonalized
simultaneously. For a common eigenvector $\Psi$, one finds 
\begin{equation}
{X}_\epsilon(p)\cdot{X}_\epsilon(p)\Psi=
\alpha (f_{p}^\epsilon)^2(v)\Psi
\end{equation}
and 
\begin{equation}
{X}_\epsilon(p)\cdot \bg{E}_\epsilon(p)\Psi
=\beta |\bg{E}{}_{\epsilon}(p)|f_{p}^\epsilon(v)\Psi.
\end{equation}
Here $\cdot$ and $\betr{\ldots}$ are with respect to the metric on SU(2).   
The values $\alpha$ and $\beta$ can obtain are determined by the representation
labels on the edges intersecting at $v$. Thus we find for the regulated area
operator 
\begin{align}
\int_S\sqrt{\Delta_\epsilon}\Psi 
&=\int_S \left[\alpha (f_{p}^\epsilon)^2(v)+2\beta
|\bg{E}{}_{\epsilon}(p)|f_{p}^\epsilon(v)+|\bg{E}{}_{\epsilon}(p)|^2
\right]^{\frac{1}{2}}\, \text{d}p\,\,\Psi\\
&=\int_S \left[\left(\sqrt{\alpha}
f_{p}^\epsilon(v)+|\bg{E}{}_{\epsilon}(p)|\right)^2 
+(\beta-\sqrt{\alpha})|\bg{E}{}_{\epsilon}(p)|f_{p}^\epsilon(v)
\right]^{\frac{1}{2}}\, \text{d}p\,\,\Psi\\
&=:\int_S \left[a_\epsilon+b_\epsilon\right]^{\frac{1}{2}}\, \text{d}p\,\,\Psi.
\end{align}
To find the limit of the integral as $\epsilon \rightarrow 0$, we use the fact
that for $a,b\geq 0$
\begin{equation}
\label{eq_ineq}
\sqrt{a}\leq\sqrt{a+b}\leq\sqrt{a}+\sqrt{b}, \qquad \sqrt{a}-\sqrt{b} \leq
\sqrt{a-b}\leq\sqrt{a},
\end{equation}
where we have additionally assumed $b\leq a$ in the second pair of inequalities.
We want to apply these inequalities to the integrand. $a_\epsilon$ is manifestly
non-negative, but we do not know the sign of $b_\epsilon$. In case it is
negative, we still know that $a_\epsilon\geq |b_\epsilon|$ by virtue of the fact
that the operator $\Delta_{\epsilon}$ is non-negative. Thus the inequalities
\eqref{eq_ineq} apply, with $a=a_\epsilon$,$b=b_\epsilon$. To finish the
calculation, we note that 
\begin{equation}
\lim_{\epsilon\rightarrow 0} \sqrt{a_\epsilon} 
=\sqrt{\alpha} \lim_{\epsilon\rightarrow 0} \int_S f_{p}^\epsilon(v)\,
\text{d}p+\lim_{\epsilon\rightarrow 0} \int_S |\bg{E}{}_{\epsilon}(p)|\,
\text{d}p
= \sqrt{\alpha} +\int_S |\bg{E}(p)|\, \text{d}p.
\end{equation}
For $\sqrt{b_\epsilon}$ on the other hand, we note that the integrand is a
product of $(f^\epsilon_p)^{1/2}$ and a bounded function. One can readily show
that the integral therefore converges to zero,
\begin{equation}
\lim_{\epsilon\rightarrow 0} \sqrt{|b_\epsilon|} =0. 
\end{equation}
Using the inequalities \eqref{eq_ineq}, we thus have
\begin{equation}
\lim_{\epsilon\rightarrow 0} \int_S\sqrt{\Delta_\epsilon}\Psi
=\lim_{\epsilon\rightarrow 0} \sqrt{a_\epsilon} 
=\sqrt{\alpha} +\int_S |\bg{E}(p)|\, \text{d}p. 
\end{equation}
The $\sqrt{\alpha}$ term is the one appearing in the vacuum representation, the
second term is precisely the area of $S$ in the background. 
\subsection{Volume}
We will now turn to the volume operator. We refer to \cite{Ashtekar:1997fb} for
motivation and
details about the chosen regularization. we will now describe the aspects of it
that are relevant for our purpose. We are considering the volume of a region
$R$. The regularization procedure consists in subdividing the region $R$ into
smaller and smaller cubes. For any given cube $C$, one defines three orthogonal
surfaces $S_{i,C}$ and the phase space function 
\begin{equation}
\label{eq_cube}
q_C[E]=\frac{1}{3!}\eta_{abc}\epsilon_{IJK}E^I_{S_a}E^J_{S_b}E^K_{S_c}
\end{equation} 
with $\eta$ the totally anti-symmetric pseudo density. 
The volume is recovered as 
\begin{equation}
\label{eq_lim}
V_R=\lim \sum_C \sqrt{|q_c|}
\end{equation}
where the limit indicated is that of the subdivision into cubes getting finer
and finer. This regularization leads the way to quantization, as $q_C$ of
\eqref{eq_cube} can immediately be quantized by replacing the classical flux
quantities by their quantum counterparts. The limit in \eqref{eq_lim} is taken
in such a way that the vertices of a graph underlying a state to be acted upon 
end up in the bulk and not on the boundary of the cubes. It turns out that once
the subdivision is fine enough, the operator action on a given state stabilizes.
Thus the limit down to cubes of coordinate-volume\footnote{Here and in the
following, we denote with \emph{coordinate volume} the volume with respect to
some fiducial metric, for example the flat metric on $\R^3$ pulled back via a
coordinate chart} zero does not have to be carried out explicitly. In an
additional step, the resulting operator is averaged in a certain way, to rid it
of a remaining dependence on the subdivision and the positions of the surfaces
inside the cubes.  

In the vacuum representation, only cubes that contain a (three or higher valent)
vertex contribute in the sum over cubes in \eqref{eq_lim}. As we will show,
this is no longer true in the new representations, and we will have to show that
this does not lead to a blow-up of the sum in the limit of finer and finer
subdivision. Indeed considering a fixed cube $C$, we will have   
\begin{equation}
q_C[E]=\frac{1}{3!}\eta_{abc}\epsilon_{IJK}
(X^I_{S_a} +\bg{E}{}^I_{S_a}) 
(X^J_{S_b}+\bg{E}{}^J_{S_b})
(X^K_{S_c}+\bg{E}{}^K_{S_c}),
\end{equation} 
thus apart from the $X^3$ term that was present in the vacuum representation, we
also have terms $X^2\bg{E}$, $X(\bg{E})^2$, and $(\bg{E})^3$. To further analyze
the situation, we consider states based on a graph $\gamma$ and assume that the
subdivision into cubes is already fine enough such that each cube in the sum in
\eqref{eq_lim} is of one of the following types (see also figure
\ref{fi_volume}): 
\begin{enumerate}
\setcounter{enumi}{-1}
\item  A cube is of \emph{type 0} if it is not intersected by any edge of the
graph $\gamma$.   
\item A cube is of \emph{type 1} if it is intersected by one edge of the graph
$\gamma$, and this intersection is in such a way that precisely one of the 
surfaces $S_{a}$ internal to the cube is intersected transversally. 
\item A cube is of \emph{type 2} if it contains precisely one vertex, and such
that precisely two of the  surfaces $S_{a}$ internal to the cube are intersected
transversally by the edges emanating from the vertex. 
\item  A cube is of \emph{type 3} if it contains precisely one vertex, and such
that all three surfaces $S_{a}$ internal to the cube are intersected
transversally by the edges emanating from the vertex. 
\end{enumerate}
\begin{figure}
\centerline{\epsfig{file=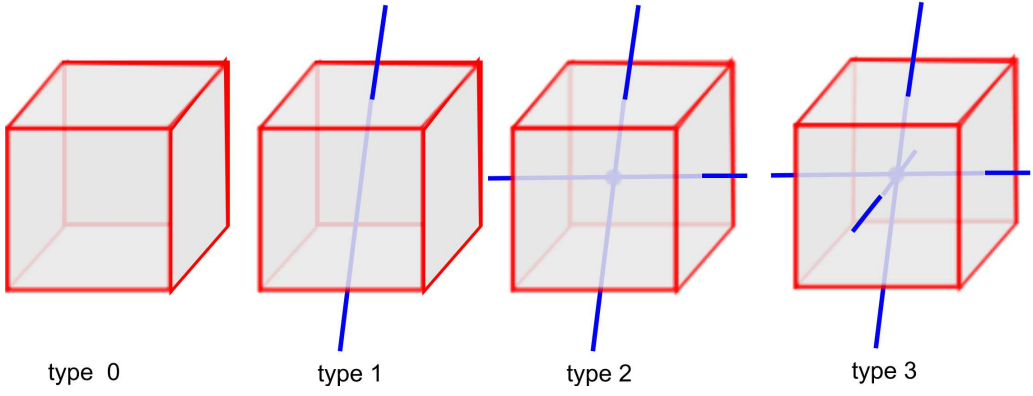, scale=0.7}}
\caption{\label{fi_volume} The different types of cubes considered in the text.}
\end{figure}
Thus we can write 
\begin{equation}
V_R=\lim \sum_{i=0}^3\sum_{C \text{ of type } i}  \sqrt{|q_c|}
\end{equation}
Let $L^3$ be the order of magnitude of the coordinate-volume of the region $R$,
and $\epsilon^3$ the typical coordinate-volume of a cube. When refining the
subdivision into cubes, the total number $N$ of cubes in the sum in
\eqref{eq_lim} thus goes like $L^3/\epsilon^3$. But the number of cubes of the
different types behaves very differently: The number of cubes of type 0 make up
the bulk and go as $L^3/\epsilon^3$. The number of cubes of type 1 goes as
$L/\epsilon$ since the edges are one-dimensional objects. Finally there are only
finitely many cubes of type 2 and 3, since there are only finitely many
vertices. 

Let us consider the action of $q_C$ for $C$ a  cube of type 2 or 3:
Schematically
\begin{equation}
q_C\Psi_\gamma=\sum \left(X^3+X^2\bg{E}
+X(\bg{E})^2+(\bg{E})^3\right)\Psi_\gamma
\end{equation}
But $\bg{E}$ goes as $\epsilon^2$, so we have 
\begin{equation}
\lim q_C\Psi_\gamma=\lim
\frac{1}{3!}\eta_{abc}\epsilon_{IJK}X^I_{S_a}X^J_{S_b}X^K_{S_c}
\end{equation}
and thus, since the sum is finite 
\begin{equation}
\lim \sum_{C \text{ of type 3,2}}\sqrt{|q_c|}\Psi_\gamma=
V'_R{}^\text{(old)}\Psi_\gamma 
\end{equation}
where $V'_R{}^{\vac}$ is the volume operator from the vacuum
representation, before the averaging to remove the remaining regularization
dependence. Next we consider a cube $C$ of type 0. We obviously have 
\begin{equation}
q_C\Psi_\gamma=\frac{1}{3!}\eta_{abc}\epsilon_{IJK}\bg{E}{}^I_{S_a}\bg{E}{}^J_{
S_b}\bg{E}{}^K_{S_c}
\end{equation} 
and hence, because of \eqref{eq_lim} 
\begin{equation}
\lim \sum_{C \text{ of type 0}}\sqrt{|q_c|}\Psi_\gamma= \bg{V}_R\Psi_\gamma 
\end{equation}
where $\bg{V}_R$ is the volume in the background geometry. Finally, we will show
that the sum over cubes of type 1 does not contribute in the limit. To that end,
we note that for $C$ of type 1, schematically 
\begin{equation}
 q_C\Psi_\gamma=\sum_{C \text{ of type 1}}
\left(X(\bg{E})^2+(\bg{E})^3\right)\Psi_\Gamma. 
\end{equation}
Since $\bg{E}$ is going like $\epsilon^2$,  we have that
$\epsilon^{-2}\sqrt{q_C}$ is converging to a finite operator. But then 
\begin{equation}
\sum_{C \text{ of type 1}}\sqrt{|q_c|}\Psi_\gamma= \epsilon^2 \sum_{C \text{ of
type 1}}
\frac{1}{\epsilon^2}\sqrt{|q_c|}\Psi_\gamma
\end{equation}
is a sum over $\approx L/\epsilon$ vectors, each with norm of order
$\epsilon^2$, and hence this sum converges to zero in the limit of $\epsilon$
going to zero. Thus we have proven that 
\begin{equation}
\lim  \sum_{C \text{ of type 3,2}}\sqrt{|q_c|}+\lim  \sum_{C \text{ of type
0}}\sqrt{|q_c|}= V'_R{}^{\vac}+\bg{V}_R.
\end{equation}
as desired. Finally, we have to carry out the averaging procedure, but this is
now trivial. The end result is a volume operator 
\begin{equation}
\lim  \sum_{C \text{ of type 3,2}}\sqrt{|q_c|}+\lim  \sum_{C \text{ of type
0}}\sqrt{|q_c|}= V_R^{\vac}+\bg{V}_R.
\end{equation}
with $V_R^{\vac}$ being the well known volume operator in the vacuum
representation. 
\section{Unitary implementation of bundle automorphisms}
\label{se_uni}
Now we consider the implementation of the automorphisms of the SU(2) principal
fiber bundle underlying the whole formalism. We will follow exactly the
treatment by Koslowski \cite{Koslowski:2007kh}. As is customarily done, we
discuss the automorphisms induced by diffeomorphisms and the gauge
transformations separately. Also in this section (and in all the following
as well), we will drop  all the $\pi$'s indicating the
representation. We hope that it is clear from context in which representation
we work. 

The first observation is that the unitary operators that implement
diffeomorphisms and gauge transformations in the vacuum representation are still
well defined and unitary in the new representations. The problem is that they do
not implement those transformations anymore. For a gauge transformation $g$ and
a diffeomorphism $\phi$ we have 
\begin{align}
U_\phi E_{S,f} U_\phi^{-1}
&=X_{\phi(S), \phi^{-1*}f}+\bg{E}_{S,f}
\neq \phi \pone E_{S,f}\\
U_g E_{S,f} U_g^{-1}
&=X_{S, \ad_g(f)}+\bg{E}_{S,f}
\neq g\gone E_{S,f}.
\end{align} 
Vice versa, it can be shown that if families of operators $\{U_g\}$ and
$\{U_\phi\}$ implement the gauge transformations and diffeomorphisms,
respectively, most of those operators can not be unitary. 

The problem is evidently that the background geometry present in the new
representations is immutable by operations on the Hilbert space and thus
prevents unitary implementations of diffeomorphisms and gauge transformations. 
This problem can be solved by enlarging the Hilbert space in such a way that
also the background geometries can be transformed, as was already observed in 
\cite{Koslowski:2007kh}. 

For a given background $\bg{E}$ we consider the Hilbert space
\begin{equation}
\label{eq_dirsum}
\mathcal{H}_{[\bg{E}]}:=\bigoplus_{\bg{\wt{E}}\in [\bg{E}]}
\mathcal{H}_{\bg{\wt{E}}}
\end{equation}
where $[\bg{E}]$ is the equivalence class of $\bg{E}$ with respect to
diffeomorphisms and gauge transformations. The Hilbert space is thus
effectively labeled by a spatial background metric modulo diffeomorphisms. 
We will write $\ket{T,\bg{\wt{E}}}$
for the cylindrical state $T$ in the (new) representation with background
geometry $\bg{\wt{E}}$ in the above Hilbert space. Thus we have 
\begin{equation}
\label{eq_scpr}
\scpr{T,\bg{E}}{T',\bg{E'}}=\scpr{T}{T'}\delta_{\bg{E},\bg{E'}}
\end{equation} 
where the first scalar product on the right hand side is just the one in the
vacuum representation. 
The Hilbert space $\mathcal{H}_{[\bg{E}]}$ carries a representation of the LQG
algebra, which is simply the direct sum of the representations on the
individual spaces. But now one can implement the diffeomorphisms and gauge
transformations unitarily, by setting \cite{Koslowski:2007kh} 
\begin{equation}
U_\phi\ket{T,\bg{E}}=\ket{\phi \pont T, \phi \pone \bg{E}}, \qquad 
U_g\ket{T,\bg{E}}=\ket{g\gont T, g \gone \bg{E}}. 
\end{equation}
We have to check that this definition gives a unitary representation, and
implements the transformations on operators correctly. That it is a
representation follows from the fact that the various actions $\triangleright$
involved are all representations. That it is unitary is immediately seen from
the fact that both, the vacuum scalar product, and the Kronecker delta used in
\eqref{eq_scpr} are invariant under gauge transformations and diffeomorphisms.

It is easy to show that with this definition, the transformations are also
correctly implemented on the operators. We give the example of the
diffeomorphisms. We calculate
\begin{align}
 U_\phi E_{S,f}\ket{T,\bg{E}}&=U_\phi
\left[\ket{X_{S,f}T,\bg{E}}+\bg{E}_{S,f}
\ket{T,\bg{E}}\right]\\
&=\ket{(\phi\pone X_{S,f}) (\phi \pont T),\phi\pone \bg{E}}
+\bg{E}_{S,f}\ket{\phi \pont T,\phi\pone \bg{E}}\\
&=\phi\pone E_{S,f} \ket{\phi \pont T,\phi\pone \bg{E}}
\end{align}
where in the last step we have used the fact that 
\begin{equation}
\phi\pone E_{S,f}\ket{\phi \pont
T,\phi\pone \bg{E}}=\ket{(\phi\pone X_{S,f}) (\phi \pont T),\phi\pone \bg{E}}
+\int_{\phi(S)}(\phi_{*}f)^I(\phi_{*}\bg{E})_I\,\, \ket{\phi \pont T,\phi\pone
\bg{E}}
\end{equation}
and that the integral simply evaluates to $\bg{E}_{S,f}$ since both, integrand
and integration region are pushed forward with $\phi$. 
From this we read off, that 
\begin{equation}
 U_\phi E_{S,f} U_\phi^{-1}=\phi\pone E_{S,f}.  
\end{equation}
For a cylindrical
function, the calculation is even simpler:
\begin{align}
 U_\phi T'\ket{T,\bg{E}}&=  U_\phi
\ket{T'T,\bg{E}}=\ket{\phi\pont(T'T),\phi\pone\bg{E}}\\
&=\ket{(\phi\pont T')(\phi\pont T)),\phi\pone\bg{E}}\\
&=(\phi\pont T')\ket{\phi\pont T,\phi\pone\bg{E}}
\end{align}
from which we read off
\begin{equation}
 U_\phi T'U_\phi^{-1}=\phi\pont T'.
\end{equation}
A similar calculation can be carried out for the gauge transformations,
showing that they, too, are implemented on operators correctly.

Finally, a straightforward calculation shows that not only the groups of gauge
transformations and spatial diffeomorphisms are correctly implemented, but also
their direct product \eqref{eq_dirpr}. If one sets $U_{(g,\phi)}:=U_g U_\phi$
one finds
\begin{equation}
 U_{(g,\phi)}U_{(g',\phi')}= U_{(g,\phi)\cdot(g',\phi')}. 
\end{equation}
We have thus shown that the action of the bundle automorphisms on the basic
variables of the theory are represented in an anomaly-free way. This in turn 
means that the hypersurface-algebra corresponding to the bundle automorphisms is
represented in an anomaly-free way.\footnote{One should not take this statement
too literally though, since -- as in the vacuum representation -- the operators
generating the unitary transformations are not well defined.} The situation
with respect to the constraint algebra is thus precisely the same as in the
vacuum representation, and we are set to tackle the implementations of the
constraints with the same methods as in the vacuum case. 

In \eqref{eq_dirsum}, the subspaces corresponding to different backgrounds are
specified to be orthogonal to each other. The reader may wonder whether a
different choice would be permissible, such that the inner product between
states with the same background is unchanged, but states with different
backgrounds are not necessarily orthogonal to each other. It is however easy to
see that the above choice is the only one that renders $E_{S,f}$ self-adjoint. 

We finally remark that the representation of the kinematical algebra of loop
quantum gravity on the Hilbert space $\mathcal{H}_{[\bg{E}]}$ is not a counter
example to the uniqueness results
\cite{Fleischhack:2004jc,Lewandowski:2005jk}, despite the fact that it carries
a unitary implementation of the bundle automorphisms. The theorems require a
cyclic representation, but the representation on $\mathcal{H}_{[\bg{E}]}$ is
highly reducible, the direct sum decomposition in \eqref{eq_dirsum} giving a
decomposition into invariant subspaces.
\section{Implementation of the constraints}
\label{se_imp}
Now we are looking at the implementation of the constraints. The natural
expectation would be that implementing the constraints using the new
representations is not possible, since they contain background geometry. But
this is not the case, as was realized in \cite{Koslowski:2007kh}. In fact, the
diffeomorphism
constraint can be implemented
exactly as in the vacuum-representation. The situation is slightly more
complicated for the Gau{\ss} constraint, as we will explain below. With
hindsight, it is not too surprising that
the implementation of the constraints succeeds. Implementing them is not about
ridding the states of background \emph{per se}, but ridding them of coordinate
and gauge dependence. 

We will first consider the implementation of the diffeomorphism constraint. Then
we will turn to the Gau{\ss} constraint. Finally we will consider states that
satisfy both classes of constraints. 
\subsection{Diffeomorphism constraint}
\label{se_diff}
The diffeomorphism constraint can be implemented on vectors in
$\mathcal{H}_{[\bg{E}]}$ with \emph{exactly} the same strategy as in the vacuum
representation. As far as we understand, this is also the strategy used in
\cite{Koslowski:2007kh}, but we will spell out more details. We will borrow
notation from \cite{Ashtekar:2004eh} and urge the reader to consult the
reference for additional information about the standard treatment. 

To start with, we make some definitions. We will call 
a diffeomorphism $\phi$ a symmetry of a triad $\bg{E}$, if $\phi$
leaves $E$ invariant, $\phi \pone \bg{E}=\bg{E}$.\footnote{This is closely
related to the concept of isometry: An isometry of the spatial geometry would be
a diffeomorphism that leaves $\bg{E}$ invariant up to a gauge transformation.
These will be  relevant in section \ref{se_imp_aut}}. Then
\begin{enumerate}
\item Let $\diff$ be the group of all diffeomorphisms. 
\item Let $\diff_{(\alpha,\bg{E})}$ be the group of diffeomorphisms that are
symmetries of $\bg{E}$ and map the graph $\alpha$ onto itself.
Note: If $\bg{E}=0$ these are just the diffeomorphisms mapping $\alpha$ onto
itself (denoted $\diff_\alpha)$ in \cite{Ashtekar:1996eg}.
\item Let $\tdiff_{(\alpha,\bg{E})}$ be the group of diffeomorphisms that are
symmetries of $\bg{E}$ and map each edge  of the graph $\alpha$ onto itself.
Note: Again, if $\bg{E}=0$ these are just the diffeomorphisms mapping each edge
of $\alpha$ onto itself, denoted as $\tdiff_\alpha$ in \cite{Ashtekar:1996eg}. 
\item Let $\symm_{(\alpha,\bg{E})}$ be the quotient
$\diff_{(\alpha,\bg{E})}/\tdiff_{(\alpha,\bg{E})}$. 
\end{enumerate}
Let $T_\alpha$ be a cylindrical function that depends on the holonomies of all
the edges of $\alpha$ in a non-trivial way. We can then define the linear form
$(T_\alpha,\bg{E}|$ on $\mathcal{H}_{[\bg{E}]}$ as follows:
\begin{equation}
\label{eq_diff}
(T_\alpha,\bg{E}|T_\beta, \bg{E'}\rangle
:=\sum_{[\phi]\in\diff/\diff_{(\alpha,\bg{E})}} 
\sscpr{T_\alpha,\bg{E}}
{P_{(\alpha,\bg{E})}^\dagger U^\dagger_\phi}
{T_\beta, \bg{E'}}
\end{equation}
where the projection $P_{(\alpha,\bg{E})}$ is defined as
\begin{equation}
\label{eq_pro}
P_{(\alpha,\bg{E})}
\ket{T_\alpha,\bg{E}}:=\frac{1}{\betr{\symm_{(\alpha,\bg{E})}}} \sum_{[\phi]\in
\symm_{(\alpha,\bg{E})}} U_\phi\ket{T_\alpha,\bg{E}}.
\end{equation}
Now we have to demonstrate that these definitions makes sense, i.e.\ that they
are well defined and give finite results, and that $(T_\alpha,\bg{E}|$ is indeed
a diffeomorphism invariant state. 

We will first address the issue of the the notions being well defined: If we 
chose, instead of $\phi$ a different member $\phi'$ of the equivalence class
$[\phi]$ on the right hand side of \eqref{eq_pro}, we will have
$\phi'=\phi\omega$, where $\omega$ is in  $\tdiff_{(\alpha,\bg{E})}$. But then 
$U_{\phi'} \ket{T_\alpha,\bg{E}}=U_\phi\ket{T_\alpha,\bg{E}}$ because $\omega$
leaves the state invariant by definition. When we chose a different
representative $\phi'$ on the right hand side of \eqref{eq_diff}, this can be
equivalently written as using the original $\phi$, but then $\theta\phi$ in
place of $\phi$ in \eqref{eq_pro}, where $\theta$ is in 
$\tdiff_{(\alpha,\bg{E})}$. But it is immediately shown, that this does not
change the result of $P_{(\alpha,\bg{E})}$. Finally, it is easy to show that,
due
to the division by the size of $\symm_{(\alpha,\bg{E})}$, the above is invariant
under subdivision of edges of $\alpha$. 

The next comment is on finiteness. We first show:
\begin{lemm}
$\symm_{(\alpha,\bg{E})}$ is finite. 
\end{lemm}
\begin{proof}
$\diff_{(\alpha,\bg{E})}$ is a subset of $\diff_\alpha$. Therefore elements of 
$\symm_{(\alpha,\bg{E})}$ will consist of elements of $\diff_\alpha$. We know
that $\diff_\alpha/\tdiff_\alpha$ is finite. We ask: Can two elements
$\phi,\phi'$ of  $\diff_{(\alpha,\bg{E})}$ that are in the same class in 
$\diff_\alpha/\tdiff_\alpha$ be in different classes in
$\symm_{(\alpha,\bg{E})}$? If the answer is `no', then $\symm_{(\alpha,\bg{E})}$
can have at most as many classes as $\diff_\alpha/\tdiff_\alpha$, and the lemma
is proven. To show that this is indeed the case, take $\phi,\phi'$ in 
$\diff_{(\alpha,\bg{E})}$ and in the same class in 
$\diff_\alpha/\tdiff_\alpha$, i.e.\
\begin{equation}
\phi=\phi'\wt{\phi}, \qquad \wt{\phi}\in \tdiff_\alpha. 
\end{equation} 
but then $\wt{\phi}=(\phi')^{-1}\phi$, and  $\diff_{(\alpha,\bg{E})}$ is a
group, so $\wt{\phi}$ is also in $\diff_{(\alpha,\bg{E})}$. Thus classes in  
$\diff_\alpha/\tdiff_\alpha$ are not split over classes in 
$\symm_{(\alpha,\bg{E})}$, and thus the lemma is proven.  
\end{proof}
Thus  $P_{(\alpha,\bg{E})}$ is a finite projection. Finally we remark that while
the sum in \eqref{eq_diff} is infinite, only finitely many terms are non-zero,
due to the fact that the sum is effectively only over diffeomorphisms that act
no-trivial on the state $\ket{T_\alpha, \bg{E}}$, and the scalar product is only
nonzero, if the backgrounds and graphs in the scalar product coincide
\emph{exactly}. 

Finally, let us demonstrate that the functional $(T_\alpha,\bg{E}|$ is indeed
diffeomorphism invariant. To this end, consider 
\begin{equation}
(T_\alpha,\bg{E}|U_{\omega}| T_\beta, \bg{E'}\rangle
\end{equation}
where $\omega$ is an arbitrary diffeomorphism. If we write out the sums, we get
an expression of the form 
\begin{equation}
\sum_{[\phi]\in\diff/\diff_{(\alpha,\bg{E})}} 
\sum_{[\phi']\in \symm_{(\alpha,\bg{E})}} 
\sscpr{\ldots}{U^\dagger_{\phi'} U^\dagger_{\phi}U_\omega}{\ldots}.
\end{equation}
But since $\diff$ is a group and $\diff_{(\alpha,\bg{E})}$ acts from the right,
we have that 
\begin{equation}
\sum_{[\phi]\in\diff/\diff_{(\alpha,\bg{E})}} U^\dagger_{\omega^{-1}\circ\phi}=
\sum_{[\phi]\in\diff/\diff_{(\alpha,\bg{E})}} U^\dagger_{\phi}
\end{equation}
and hence 
\begin{equation}
(T_\alpha,\bg{E}|U_{\omega}| T_\beta, \bg{E'}\rangle
=(T_\alpha,\bg{E}| T_\beta, \bg{E'}\rangle
\end{equation}
as it should be. 

The definition of the inner product on the diffeomorphism invariant states now
proceeds exactly as in the vacuum case, and leads to a diffeomorphism invariant
Hilbert space $\mathcal{H}_{[\bg E]}^{\diff}$.  We emphasize that the
construction is such that for the fully degenerate background geometry
$\bg{E}=0$, we have 
$\mathcal{H}_{[0]}^{\diff}=\mathcal{H}^{\diff}$, where the latter is the
(standard) diffeomorphism-invariant Hilbert space of the vacuum representation.
\subsection{Gau{\ss} constraint}
Now we come to the implementation of the Gau{\ss} constraint. Here our treatment
will differ somewhat from \cite{Koslowski:2007kh}. We will comment on the
differences below. 

Also, we find that we cannot blindly follow the strategy used in the vacuum
representation: There, implementation of the constraint can be done by group
averaging involving an \emph{integral} over the group of gauge transformations. 
While this group is huge, since there are only finitely many vertices involved
in any given graph, the integral boils down to finitely many integrals over
SU(2) and is thus well defined. In our context, since gauge transformation
may act non-trivially on all of $\bg{E}$, it is unclear how to extend this
treatment
to the new representations. Instead one might be tempted to follow the strategy
that was used in the case of the diffeomorphism constraint. There, a \emph{sum}
over the group elements was performed. Now, since a gauge transform of a state
may not be orthogonal to the state itself in the vacuum representation, one
might suspect that this strategy runs into problems with convergence. Indeed,
this is the case. To see this, we try to define
\begin{equation}
\label{eq_large}
 \{T, \bg{E}| =\sum_{g\in \gauge/\gauge_{T,\bg{E}}}\bra{T,\bg{E}}U^\dagger_g  
\end{equation}
where $\gauge_{T,\bg{E}}$ are the gauge transformation that are a
symmetry of $\bg{E}$ (i.e. leave $\bg{E}$ invariant) and leave $T$ invariant.
The problem with this is that the
sum may contain infinitely many $g$ which change $T$ but leave $\bg{E}$
invariant. Think for example of gauge transformations that are rotations with
axis $\bg{E}$ at each point. Such transformations will still generically change
a given functional $T$. When evaluating $\{T, \bg{E}|$ on a state 
$\ket{T',\bg{E}}$ with $T'$ gauge invariant and $\scpr{T}{T'}\neq 0$, each
transformation that fixes $\bg{E}$ but changes $T$ will produce the same
non-vanishing number, leading to a divergence.    

The problem described above can be avoided if one starts with gauge invariant
functionals $T$. Then obviously there are no gauge transformations that change
$T$ and leave $\bg{E}$ invariant. Let $\gauge_\bg{E}$ be the group of gauge
transformations that are symmetries of $\bg{E}$. For a gauge invariant $T$ we
set
\begin{equation}
\label{eq_small}
\{T, \bg{E}| :=\sum_{g\in \gauge/\gauge_{\bg{E}}}\bra{T,\bg{E}}U^\dagger_g.
\end{equation}
This is obviously gauge invariant, and evaluated on a state there can be at most
one non-zero term:
\begin{equation}
\label{eq_gscpr}
\{T, \bg{E}|T',\bg{E'}\rangle = \scpr{T}{T'}\begin{cases} 1 &\text{ if }
\bg{E},\bg{E'} \text{ are gauge-related}\\
0& \text{ otherwise}
\end{cases}
\end{equation}
We emphasize that the
construction is such that for the fully degenerate background geometry
$\bg{E}=0$, we have 
$\mathcal{H}_{[0]}^{\gauge}=\mathcal{H}^{\gauge}$, where the latter is the
(standard) gauge-invariant Hilbert space of the vacuum representation.

We also note that (in view of \eqref{eq_gscpr}), an orthonormal  basis of  
$\mathcal{H}_{[\bg{E}]}^{\gauge}$ can be given as follows: First, since $|T, \bg{E}\}$ depends on $\bg{E}$ only through the gauge equivalence class $[\bg{E}]$, we might as well denote it by $|T, [\bg{E}]\}$.
Then a  basis is given by  $\{\,|T, [\bg{E}]\}\,\}$, where $T$ runs through all spin nets, and $[\bg{E}]$ through all gauge equivalence classes.

We now come to the difference with \cite{Koslowski:2007kh}. Koslowski
puts forward the idea, that one should also consider
states in which the connection is coupled to the background, as these can be
gauge invariant, too. Take for example the case of a loop
$\alpha$ with endpoint $p$, and consider the functional 
\begin{equation}
 T[A,\bg{E}] =\tr\left(h_\alpha[A]\bg{E}(p)\right).  
\end{equation}
This is invariant, provided both $A$ and $\bg{E}$  transform under gauge
transformations. But we do not see how this can be achieved in the present
formalism: Nothing prevents us from defining the state $\ket{T[\,\cdot\, ,
\bg{E}],\bg{E}}$, where $T$ is the functional defined above, but it is not gauge
invariant. Rather 
\begin{equation}
 U_g\ket{T[\,\cdot\, , \bg{E}],\bg{E}}
=\ket{T[g^{-1} \,\cdot\, g,\bg{E}],g\gone \bg{E}}
=\ket{T[\, \cdot\, , g \gone \bg{E}],g \gone \bg{E}}. 
\end{equation}
Instead, one might consider it more promising to start with the
functional (a treatment along these lines was probably implied  
in \cite{Koslowski:2007kh} )
\begin{equation}
\label{eq_ginv}
 \{\psi|:=\sum_{g\in\mathcal{G}/\mathcal{G}_{0,\bg{E}}}\bra{T[\,\cdot\, , \bg{E}],g \gone \bg{E}}. 
\end{equation}
Here, $0$ denotes the zero spin network. One can envision this functional as
being obtained from the action of a hypothetical operator $\tr(h_\alpha
\bg{\widehat{E}}(p))$ on the state 
\begin{equation}
 \{0,[\bg{E}]|  \equiv\sum_{g\in\mathcal{G}/\mathcal{G}_{\bg{E}}}\bra{0,g
\gone \bg{E}}.
\end{equation}
Indeed, operators of this type have
been considered in the early days of loop quantum gravity. One can check that 
\eqref{eq_ginv} defines functional on states, 
\begin{equation}
 \label{eq_newfun}
 \{\psi\ket{T,\bg{E'}}
=\begin{cases}
 \scpr{T[\,\cdot\, , \bg{E'}]}{T} & \text{ if } \bg{E'}\in [E]\\
0 & \text{ otherwise} 
 \end{cases}.
\end{equation}
A straightforward calculation shows that this functional is also gauge
invariant. It is, however, not an element of the diffeomorphism invariant Hilbert space. This can be seen very easily as follows: While $\{\psi|$ evaluates to something non-zero only on states $\ket{T,\bg{E'}}$ where $T$ is not gauge invariant, all functionals in $\mathcal{H}_{[\bg{E}]}^{\gauge}$ are non-zero only on states $\ket{T,\bg{E'}}$ where $T$ \emph{is} gauge invariant. 

Thus, to give functionals such as \eqref{eq_newfun} a home, one would need a scalar product on them. This may be possible, for example by modifying the procedure suggested in \eqref{eq_large} in such a way that the sum is over a smaller set. This however risks to yield a non-gauge-invariant functional.  
To summarize, it may well possible -- and this possibility is very exciting -- to also introduce states that couple to the background, but this would necessitate a departure from the framework considered here (and as far as we understand, also from the one in \cite{Koslowski:2007kh}).
\subsection{States invariant under bundle automorphisms}
\label{se_imp_aut}
Up to now we have considered diffeomorphisms and gauge transformations
separately. Now we want to show how one can obtain states that are invariant
under their semidirect product, the bundle automorphisms. The idea is to first
average over gauge transformations, then over diffeomorphisms. One has to be
careful, however because there may be diffeomorphisms that are symmetries up to
a gauge transformation. Applying these to a gauge invariant functional and
summing over them leads to an over-counting and, possibly, to a divergence.    

Thus the averaging for the diffeomorphism invariant states has to be modified
slightly. We denote by $\giff_{(\alpha,\bg{E})}$ the group of diffeomorphisms
that are
symmetries of $\bg{E}$ up to a gauge transformation, and map the graph $\alpha$
onto itself.
Let $\tgiff_{(\alpha,\bg{E})}$ be the group of diffeomorphisms that are
symmetries of $\bg{E}$ up to gauge transformations and map each edge  of the
graph $\alpha$ onto itself.
Finally  Let $\gymm_{(\alpha,\bg{E})}$ be the quotient of the above. Then, for a
state $\ket{T, \bg{E}}$ with $T$ gauge invariant, the group averaging over
diffeomorphisms and gauge transformations is 
\begin{equation}
(T_\alpha,\bg{E}|T_\beta, \bg{E'}\rangle
:=\sum_{[\phi]\in\diff/\giff_{(\alpha,\bg{E})}} 
\{T_\alpha,\bg{E}
|P_{(\alpha,\bg{E})}^\dagger U^\dagger_\phi|
T_\beta, \bg{E'}\rangle
\end{equation}
where
\begin{equation}
P_{(\alpha,\bg{E}}
\ket{T_\alpha,\bg{E})}:=\frac{1}{\betr{\gymm_{(\alpha,\bg{E})}}} \sum_{[\phi]\in
\gymm_{(\alpha,\bg{E})}} U_\phi\ket{T_\alpha,\bg{E}}.
\end{equation}
It is easy to show that this is well defined, with the same reasoning as given
in section \ref{se_diff}. The proof that $\gymm$ is finite goes through exactly
as before. It is also easy to see that evaluation is always finite. 
Finally it is obvious that the resulting functional is invariant under both,
diffeomorphisms and gauge transformations.   
\section{Outlook}
\label{se_con}
In the present article, we have demonstrated how the new representations form a
family with the (standard) vacuum representation being one member. This begs
the question as to the significance and use of these representations.

As we have
already stated before, we do not think that the new representations are as
fundamental as the vacuum representation, because they contain so much
classical background structure. Still, it may be interesting to see in what
sense they can be used to solve the Hamiltonian constraint. On the one hand,
one could take over Thiemann's construction literally, and see whether the
constraint is well defined, and how solutions would look. On the other hand
the new representations allow for new quantization strategies, because the
volume operator can now have non-vanishing kernel. Likewise, we have
preliminary indications that one may be able to define an operator
corresponding to the integral of curvature over a surface. 

All of this would also be interesting for applying the new representation in
situations, in which one would like to treat most of the gravitational
excitations as background, and only study in detail the excitations of geometry
over this kind of geometric condensate. 

Finally, the new representations could be a spring board to the construction of
yet other representations which are more semiclassical.

On a more technical level, it may be interesting to study different types of
backgrounds. When we were discussing the backgrounds $\bg{E}$ the reader may
have had a smooth, non-degenerate geometry in mind. But more general choices are
also possible, as long as fluxes, areas and volumes are still well defined.
Possibilities include for example fields that are only non-vanishing on
submanifolds.   

We hope to come back to some of the questions raised above, at another time. 

\section*{Acknowledgements}
I thank T.\ Koslowski for interesting and helpful comments on the subject and on a draft of this work. The anonymous referees also provided interesting comments and suggestions that helped improve the manuscript. This work has been partially supported by the Spanish MICINN
project No. FIS2008-06078-C03-03. I also gratefully acknowledge a travel grant from the European Science Foundation research network \emph{Quantum Geometry and Quantum Gravity}.  


\end{document}